\documentclass[a4paper,envcountsame,envcountsect,runningheads]{lmcs}
\pdfoutput=1
\usepackage[utf8]{inputenc}

\usepackage{lastpage}
\lmcsdoi{20}{3}{11}
\lmcsheading{}{\pageref{LastPage}}{}{}%
{Mar.~01,~2023}{Aug.~05,~2024}{}


\usepackage{graphicx}
\usepackage{amsmath}
\usepackage{amssymb}
\usepackage{pstricks}
\usepackage{pst-node}
\usepackage{pst-plot}
\usepackage{uapmi}
\usepackage{color}
\usepackage{bm}
\usepackage{algpseudocode}
\usepackage{algorithm}
\usepackage{algpseudocode}
\usepackage{float}

\renewcommand{\maps}{\rightarrow}

\newcommand{\defiff}{\Leftrightarrow_{def}}

\newcommand{\bez}{\backslash}

\newcommand{\st}{\,|\,}
\newcommand{\setsystem}[1]{\mathbf{#1}}

\newcommand{\ord}{\leq}
\newcommand{\nord}{\nleq}
\newcommand{\revord}{\geq}
\newcommand{\strord}{<}
\newcommand{\strrevord}{>}

\newcommand{\ensemble}[1]{#1^{\{\}}}

\newcommand{\fixpoints}[1]{\keyw{fix}(#1)}

\newcommand{\comparable}{\lessgtr}
\newcommand{\incomparable}{\not\lessgtr}
\renewcommand{\equiv}{\sim}
\newcommand{\equirel}{\equiv}
\newcommand{\quot}{/\hspace*{-1mm}}
\newcommand{\maj}[2]{\keyw{maj}(#1,#2)}
\newcommand{\least}[1]{\keyw{lst}(#1)}
\newcommand{\mymax}[1]{\max{(#1)}}
\newcommand{\mymin}[1]{\min{(#1)}}
\newcommand{\bigo}{\mathcal{O}}
\newcommand{\csystem}[1]{\mathcal{C}(#1)}
\newcommand{\pcsystem}[1]{\mathcal{PC}(#1)}
\newcommand{\equiclass}[2]{[#1]_{#2}}

\begin{document}

\title[Isolated Suborders and Counting Closure Operators]{Isolated Suborders and their Application to Counting Closure Operators}
\author[R.~Gl\"{u}ck]{Roland Gl\"{u}ck\lmcsorcid{0000-0001-7909-1942}}
\address{Deutsches Zentrum f\"{u}r Luft- und Raumfahrt, D-86159 Augsburg, Germany}
\email{roland.glueck@dlr.de}
 
\begin{abstract}
In this paper, we investigate the interplay between isolated suborders and closures.
Isolated suborders are a special kind of suborders and can be used to diminish the number of elements of an ordered set by means of a quotient construction.
The decisive point is that there are simple formulae establishing relationships between the number of closures in the original ordered set and the quotient thereof induced by isolated suborders.
We show how these connections can be used to derive a recursive algorithm for counting closures, provided the ordered set under consideration contains suitable isolated suborders.
\end{abstract}

\maketitle

\section{Introduction}\label{sec:introduction}
A widespread and common concept in various areas of mathematics and computer science are hull or closure operators, i.e., idempotent, isotone and extensive endofunctions on some ordered set.
The best-known examples include the topological closure in traditional analysis, the (reflexive) transitive closure of a relation or a graph, and the Kleene closure in language theory.
There are also more complicated and sophisticated appearances as for example in automated reasoning (see e.g.~\cite{ElloumiBJSOF14}), data base theory (as in~\cite{DemetrovicsHLM92}), or the algebraic analysis of connected components (as demonstrated in~\cite{RelMiCSGluck17}).
Most of this work uses hulls or closures mainly as a tool for specific purposes but does not investigate the properties of these operators.

If we take a look at the work dealing with actual properties of closures, then we see that most of this work is concerned with closures on powerset lattices which falls also under the term ``Moore family'' (here ~\cite{CaspardM03} gives a survey; for more recent results see e.g. ~\cite{BrinkmannD18}).
Some other work deals with closures on lattices although closures can also be defined on general ordered sets.
The appoach from~\cite{DemetrovicsHLM92} does not rely on lattice properties (it is not explicitly stated whether a lattice structure is required); however, if one tries to extend the ideas from one single database table to several tables, then a nearby approach (using tuples of choice functions) leads to an order lacking at least a binary supremum operator.

In recent years, counting structures of interest has become a rising area of research in order theory and related topics.
For example, \cite{AlpayJS21} gives numbers of so-called $d\ell$-structures of various kinds,~\cite{QuinteroRRV20} counts join-endomorphisms in lattices,~\cite{BonzioBV18} generates and counts a certain kind of bisemilattices whereas the topic of~\cite{BerghammerBW21} are topological spaces, and different kinds of posets are counted in~\cite{FahrenbergJSTRAMiCS2020}.
However, to our best knowledge, there is no work dealing with the number of closures on general lattices or orders.
The only results we are aware of concern the power set lattice $(\mathcal{P}(S),\subseteq)$.
For this special structure, the exact number of closure operators is known only up to $|S|=7$ (for curiosity, there are 14.087.648.235.707.352.472 of them, as shown in~\cite{ColombIR10} only in 2010).

The present work introduces a heuristic method for structuring ordered sets in a way that eases under certain circumstances the computation of the number of closure operators.
The key idea are so called isolated suborders which are intuitively speaking suborders which have contact with the rest of the ordered set only via their least and greatest element.
By means of quotient orders, we can reduce the number of elements of the ordered set under consideration and obtain an ordered set with a certain structure for which there are closed formulae available for the number of closure operators.
Similar ideas were used already in the predecessor work of~\cite{GlueckRAMiCS21}.
There, however, the area of application was restricted to lattices whereas the present work generalizes the results to general ordered sets.
We stress that this generalization does not lead to the common loss of strength of claims associated with generalization.
In particular, the specialization of the results in the present paper to lattices leads exactly to the results obtained in~\cite{GlueckRAMiCS21} for lattices.
This holds also for the main result of the paper, Algorithm~\ref{alg:counting}, which retains the same structure as in the predecessor paper.
Additionally to this generalization, the present work contains also thoughts about the computation of isolated suborders.

Of course, heuristic approaches defy often precise analysis, nevertheless, heuristics are an established tool if one has to deal with computationally hard problems.
One of the oldest examples is the $\text{A}^*$ algorithm, first described 1968 in~\cite{AStar}.
More recent examples are algorithms based on tree decomposition as in~\cite{BodlaenderDynamic}, outerplanarity as in~\cite{KammerOuterplanarity}, or the vast field of genetic algorithms, see e.g.~\cite{GoldbergGenetic} for a first introduction.
All of these approaches have in common that a rigorous analysis is hard or even impossible; consequently, we will also give only a rough analysis of our algorithm in Section~\ref{sec:counting}.

The remainder of this work is organized as follows: in Section~\ref{sec:basics}, we introduce some notation we will use in the sequel (on other places, we introduce notation ad hoc to spare the reader annoying look-ups).
Section~\ref{sec:closures} introduces the main topic of the present work, closures, and shows important relations between different characterizations of closures.
The main tool of the present work, isolated suborders, are the topic of Section~\ref{sec:ISOs} whereas Section~\ref{sec:isoAndCloSys} studies the interplay between closures and isolated suborders.
Because we look at isolated suborders not only as an object of study but will use them as tool in an algorithm, we investigate in Section~\ref{sec:compISO} how to compute isolated suborders.
In Section~\ref{sec:counting}, we put all our results together to obtain an algorithm for counting closures which can make use of favorable structures of an ordered set under consideration.
Finally, the concluding Section~\ref{sec:conclusion} gives a short retrospect of the work, raises some open issues, and sketches ways of further research.

\section{Basic Notions and Properties}\label{sec:basics}

In this paper, we presuppose general knowledge of order and lattice theory, and refer e.g. to~\cite{DaveyP,Graetzer2011,RomanLattices} for the basics and e.g. to~\cite{Birkhoff67,VarLatt} for more advanced topics. 
We will, however, recapitulate some topics in order to clarify matters and to obtain a consistent notion.

The symbol $\ord$ and derived variants thereof (by indexes or primes) denote always an order whose carrier set is usually denoted by $S$ and variants thereof.
Given an order $\ord$, we use the symbol $\nord$ for the relation $\nord\defiff\neg(x\ord y)$.
The associated strict order, reverse order, and strict reverse order are denoted by $\strord$, $\revord$, and $\strrevord$, respectively.
In the case of existence, $\bot$ and $\top$ (also possibly indexed) stand for the least and greatest element of an order.
The sets of maximal and minimal elements of a subset $S'\subseteq S$ are denoted by $\mymax{S'}$ and $\mymin{S'}$ , respectively, whereas $\least{S'}$ denotes the (possibly empty) set of least elements of $S'$.
An element $x$ is said to \emph{majorize} an element $y$ if $x\revord y$ holds, and we extend this concept to sets by $\maj{x}{S'}\defeqs\{y\in S'\st y\revord x\}$.
Two elements $x$ and $y$ are called \emph{comparable}, written $x\comparable y$, if $x\ord y$ or $y\ord x$ holds.
Consequently, we call $x$ and $y$ \emph{incomparable} if they are not comparable, and denote this by $x\incomparable y$.
As usual, a \emph{chain} is a set of elements which are pairwise comparable.
A subset $S'\subseteq S$ of an ordered set $(S,\ord)$ is called \emph{convex} if for all $x,y\in S'$ and all $z\in S$ the implication $x\ord z\ord y\Rightarrow z\in S'$ holds.
For intervals, we use the common notations $[a,b]\defeqs\{x\st a\ord x\wedge x\ord b\}$ and $]a,b]\defeqs[a,b]\bez\{a\}$.
For a relation $R\subseteq S\times S$ and a subset $S'\subseteq S$, we define the \emph{restriction} of $R$ to $S'$ routinely by $R|_{S'}\defeqs\{(s',t')\in R\st(s',t')\in S'\times S'\}$.

Given an equivalence relation $E\subseteq S\times S$, we denote the equivalence class of an element $s\in S$ under $E$ by $\equiclass{s}{E}$.
For the set of equivalence classes of such an equivalence relation $E\subseteq S\times S$ we use the notation $S/E$.
For an arbitrary relation $R\subseteq S\times S$ and an equivalence relation $E\subseteq S\times S$ we define the \emph{quotient} $R/E\subseteq S/E\times S/E$ by $(\equiclass{x}{E},\equiclass{y}{E})\in R/E\defiff$ $\exists x'\in\equiclass{x}{E}\exists y'\in\equiclass{y}{E}:(x',y')\in R$.
If $(S,\ord)$ is an ordered set and $E\subseteq S\times S$ is an equivalence relation such that $(S/E,\ord/E)$ is an ordered set, then we say that $E$ is \emph{order generating}.
In this case, we may write also $\ord_{E}$ instead of $\ord/E$.
Clearly, if under these circumstances both $\equiclass{x}{E}$ and $\equiclass{y}{E}$ are singleton sets, then we have the equivalence $x\ord y\iff\equiclass{x}{E}\ord_{E}\equiclass{y}{E}$ to which we will often refer as \emph{homomorphism properties}.
Trivial examples of order generating equivalences are the identity and the universal relation.
However, as counterexample, the relation $\equiv$ on the ordered set $(\ZZ,\ord)$ (where $\ord$ denotes the usual ordering of the integers), defined by $x\equiv y\defiff x\cdot y\neq 0\vee x=0=y$, is obviously not order generating.

Since we will often have to consider the union of the sets of a set system $\setsystem{C}$ (i.e., $\setsystem{C}$ is a set of sets), we use the abbreviation $\bigcup\setsystem{C}\defeqs\bigcup\limits_{C\in\setsystem{C}}C$ to make the text easier to read.
Conversely, dealing with a set $C$, we use $\ensemble{C}\defeqs\{\{c\}\st c\in C\}$ to denote the set of singleton sets consisting of the elements of $C$.

\section{Closures}\label{sec:closures}

The main topic of this work, closures, can be characterized in two different ways, namely as endofunctions on ordered sets or as subsets of ordered sets.
At the end of this section we will see that these characterizations are cryptomorphic; we start with the functional definition:

\begin{defi}\label{def:closureFunction}
 Given an ordered set $(S,\ord)$, an endofunction $c$ on $S$ is called a \emph{closure operator} if it fulfills the following properties for all $x,y\in S$:
 \begin{enumerate}
  \item $x\ord c(x)$\hfill(extensitivity)
  \item $x\ord y\Rightarrow c(x)\ord c(y)$\hfill(isotony)
  \item $c(c(x))=c(x)$\hfill(idempotence)
 \end{enumerate}
\end{defi}

A useful easy consequence of this definition is the following corollary:

\begin{cor}\label{cor:closExpansII}
 Let $c$ be a closure operator and $x,y\in S$ elements such that $y\revord c(x)$ ($y=c(x)$) holds. Then $x\ord y$ holds.
\end{cor}

\begin{proof}
This follows simply from $x\ord c(x)$ and transitivity of $\ord$.
\end{proof}

The next definition is a characterization of closures as subsets of ordered sets:

\begin{defi}\label{def:closureSystem}
 Given an ordered set $(S,\ord)$ a subset $C\subseteq S$ is called a \emph{closure system} if for every $s\in S$ the set $\maj{s}{C}$ has a least element.
\end{defi}

The set of all closure systems of an ordered set $(S,\ord)$ is denoted by $\csystem{S}$ (here we assume that the order on $S$ is clear from context).

\begin{rem}
In the literature, e.g.~\cite{Graetzer2011}, one finds a more complicated version of Definition~\ref{def:closureSystem} which imposes also the requirement that a closure system is closed under binary infima. 
However, a reviewer of~\cite{GlueckRAMiCS21} pointed out that this requirement is redundant.
This observation lead in the consequence to the generalization of~\cite{GlueckRAMiCS21} from lattices to the present form handling general ordered sets.
This is a rare case where generalization leads also to simplification due to the more concise formulation of Definition~\ref{def:closureSystem}.
\end{rem}

It turns out that closure operators generate closure systems:

\begin{lem}\label{lem:closFuncIndClosSys}
 For every closure operator $c$, the set $\fixpoints{c}$ of fixpoints of $c$ is a closure system.
\end{lem}

\begin{proof}
Let $s\in S$ be an arbitrary element.
Clearly, $c(s)\in\fixpoints{c}$ holds, and due to extensitivity of $c$ we have $c(s)\in\maj{s}{\fixpoints{c}}$.
Let us now pick an arbitrary $f\in\fixpoints{c}$ with the property $s\ord f$.
Now isotony of $c$ and $f\in\fixpoints{c}$ imply $c(s)\ord f$ which completes the proof.
\end{proof}

Also, closure systems determine closure operators in a unique way:

\begin{lem}\label{lem:closSysIndClosFunf}
 Let $C$ be a closure system. Then there exists exactly one closure operator $c$ with $\fixpoints{c}=C$.
\end{lem}

\begin{proof}
We define the function $c$ by $c(s)\defeqs\least{\maj{s}{C}}$ (note that this is well-defined due to the properties of a closure system according to Definition~\ref{def:closureSystem}).
Let us first check that $c$ is indeed a closure operator:
\begin{itemize}
 \item Extensivity: this is obvious since every element is mapped to a majorizing one.
 \item Isotony: under the assumption $s\ord t$ we have $\maj{s}{C}\supseteq\maj{t}{C}$ and hence $c(s)=\least{\maj{s}{C}}\ord\least{\maj{t}{C}}=c(t)$.
 \item Idempotence: by already shown extensivity we have $s\ord c(s)$ and hence $\maj{s}{C}\supseteq\maj{c(s)}{C}$. Let us now pick an arbitrary $s'\in\maj{s}{C}$. Because $c(s)$ is the least element of $\maj{s}{C}$ we have $c(s)\ord s'$, hence $s'\in\maj{c(s)}{C}$, implying $\maj{s}{C}=\maj{c(s)}{C}$. Now $c(s)=c(c(s))$ follows from this set equality and definition of $c$.
\end{itemize}
Let us now assume that there is another closure operator $c'$ with $\fixpoints{c'}=C$, and pick an arbitrary $s\in S$.
Because $c'$ is idempotent, we have $c'(s)\in C$ and hence $c(c'(s))=c(s)$ from where we conclude that $c'(s)\ord c(s)$ holds (this is due to Corollary~\ref{cor:closExpansII}).
Symmetrically, we obtain $c(s)\ord c'(s)$ and hence $c(s)=c'(s)$ which shows uniqueness of $c$.
\end{proof}

Lemmata~\ref{lem:closFuncIndClosSys} and~\ref{lem:closSysIndClosFunf} establish a cryptomorphic one-to-one correspondence between closure operators and closure systems on an ordered set.
Since the main contribution of this work deals with counting of closures (both functions and systems), it is sufficient to use the more convenient characterization.
In this case, closure systems are much easier to handle than closure operators.

\section{Isolated Suborders}\label{sec:ISOs}

The main tool for structuring ordered sets we will use is the subject of the following definition:

\begin{defi}\label{def:isolatedSuborder}
Let $(S,\ord)$ be an ordered set. 
A subset $S'\subseteq S$ is called an \emph{isolated suborder} if it fulfills the following properties:
\begin{enumerate}
 \item $S'$ has a greatest element $\top_{S'}$ and least element $\bot_{S'}$.
 \item $\forall x\notin S'\forall y'\in S':y'\ord x\Rightarrow\top_{S'}\ord x$
 \item $\forall x\notin S'\forall y'\in S':x\ord y'\Rightarrow x\ord\bot_{S'}$
\end{enumerate}
\end{defi}

Intuitively, an isolated suborder $S'$ can be ``entered from below'' only via $\bot_{S'}$ and ``left upwards'' only via $\top_{S'}$.
We call an isolated suborder \emph{nontrivial} if $S'$ does not equal $S$.
A \emph{summit isolated suborder} is a suborder $S'$ such that $\top_{S'}\in\mymax{S}$ holds.
If $|S'|>1$ holds we call an isolated suborder \emph{nonsingleton}, and a \emph{useful} isolated suborder is a nontrivial non-singleton isolated suborder.

Another property we will need to make our ideas work is that an order does not ``branch upwards'' at an element under consideration (in our case at its top element):

\begin{defi}\label{def:bottleneck}
 Given an ordered set $(S,\ord)$ we call an element $b\in S$ a \emph{bottleneck} of an element $x\in S$ if the following conditions are fulfilled:
 \begin{enumerate}
  \item $b\strrevord x$,
  \item $[x,b]$ is a chain, and
  \item for all $y\in S$, the implication $y\strrevord x\Rightarrow(y\in[x,b]\vee y\strrevord b)$ holds.
 \end{enumerate}
\end{defi}

Consequently, an \emph{isolated suborder with bottleneck} is an isolated suborder $S'$ such that $\top_{S'}$ has a bottleneck.
Obviously, given an element $x$ with a bottleneck $b$, every element in $]x,b]$ is also a bottleneck of $x$.

Figure~\ref{fig:Suborders} illustrates these definitions.
At the left, an isolated suborder with bottleneck is shown.
In the middle, we find an isolated suborder without bottleneck, and at the right, a summit isolated suborder is given.
Note that neither the overall order nor the isolated suborder with bottleneck are lattices (in contrast, the other two isolated suborders are indeed lattices).

\begin{figure}
 \includegraphics[width=1.1\textwidth]{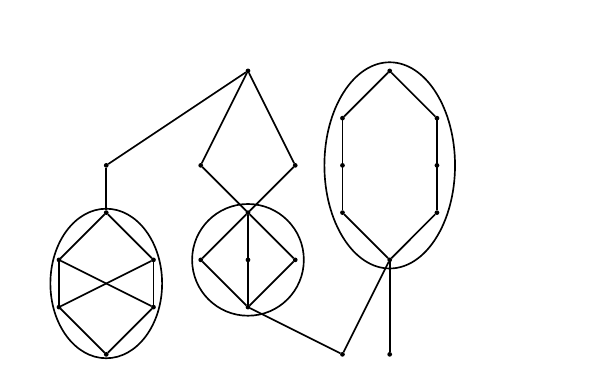}
 \caption{Various Kinds of Isolated Suborders}
\label{fig:Suborders}
\end{figure}

We will use isolated suborders to derive quotients of an order.
To this end, we define for an isolated suborder $S'\subseteq S$ an equivalence relation $\equiv_{S'}$ by $x\equiv_{S'}y\defiff x=y\vee(x\in S'\wedge y\in S')$.
Clearly, for an element $x\notin S'$, the equivalence class $\equiclass{x}{\equiv_{S'}}$ is the singleton set $\{x\}$, and for all elements $x\in S'$ the equivalence class $\equiclass{x}{\equiv_{S'}}$ coincides with $S'$.
A crucial point is that $\equiv_{S'}$ is even order generating:

\begin{lem}\label{lem:suborderOrderGenerating}
Let $S'$ be an isolated suborder of an ordered set $(S,\ord)$. Then $\equiv_{S'}$ is order generating.
\end{lem}

\begin{proof} It is easy to see that $\ord/\equiv_{S'}$ (we use this notation since we do not know yet whether $\equiv_{S'}$ is order generating) is both reflexive and transitive so it remains to show that it is also antisymmetric. 
To this end, we pick two arbitrary $\equiclass{s}{\equiv_{S'}},\equiclass{t}{\equiv_{S'}}\in S/\equiv_{S'}$ such that both $\equiclass{s}{\equiv_{S'}}\ord/\equiv_{S'}\equiclass{t}{\equiv_{S'}}$ and $\equiclass{t}{\equiv_{S'}}\ord/\equiv_{S'}\equiclass{s}{\equiv_{S'}}$ hold.
In the case $s,t\notin S'$, we have $\equiclass{s}{\equirel_{S'}}=\{s\}$ and $\equiclass{t}{\equirel_{S'}}=\{t\}$, and hence both $s\ord t$ and $t\ord s$ by homomorphism properties.
Now $\equiclass{s}{\equirel_{S'}}=\equiclass{t}{\equirel_{S'}}$ is an easy consequence of the antisymmetry of $\ord$.
If $s,t\in S'$ holds, then we have $\equiclass{s}{\equirel_{S'}}=\equiclass{t}{\equirel_{S'}}$ by construction of $\equirel_{S'}$.
For the last case we assume w.l.o.g. that $s\in S'$ and $t\notin S'$ hold.
From $\equiclass{s}{\equiv_{S'}}\ord/\equiv_{S'}\equiclass{t}{\equirel_{S'}}$ we conclude that there is an $s_1\in S'$ with $s_1\ord t$ (note that $\equiclass{t}{\equirel_{S'}}$ is the singleton set $\{t\}$).
By definition of an isolated suborder, this implies $\top_{S'}\ord t$, and symmetrically we obtain $t\ord\bot_{S'}$.
This leads to the chain $t\ord\bot_{S'}\ord\top_{S'}\ord t$, implying among other things $t=\bot_{S'}$ and hence $t\in S'$, contradicting the choice of $t$.
This finishes the proof since it shows that the last case can not occur.
\end{proof}

This lemma justifies the writing $\ord_{\equirel_{S'}}$ which we will use from now on.

An important property of isolated suborders is their convexity:

\begin{lem}\label{lem:suborderConvex}
 Let $S'$ be an isolated suborder of an ordered set $(S,\ord)$. Then $S'$ is convex.
\end{lem}

\begin{proof} Let $s,t\in S'$ be arbitrary elements of $S'$ with $s\ord t$, and assume that there is an $u\notin S'$ such that $s\ord u\ord t$ holds.
By $\bot_{S'}\ord s\ord u$, we have $\bot_{S'}\ord u$.
On the other hand, from $u\ord t$ and $u\notin S'$ we obtain $u\ord\bot_{S'}$ by definition of an isolated suborder; so altogether we have $u=\bot_{S'}$.
However, this is a contradiction to the assumption $u\notin S'$.
\end{proof}

In particular, this means that an isolated suborder $S'$ is the same as the interval $[\bot_{S'},\top_{S'}]$.
However, not every interval is an isolated suborder: consider as counterexample the interval $[\emptyset,\{1\}]$ in the ordered set $(\mathcal{P}(\{1,2\}),\subseteq)$.
In this setting, we have $\emptyset\subseteq\{2\}$ but not $\{1\}\subseteq\{2\}$, so the second part of Definition~\ref{def:isolatedSuborder} is not fulfilled.

The next lemma states intuitively spoken that isolated suborders with common elements can not lie side by side:

\begin{lem}\label{lem:chainBotTop}
 Let $S_1$ and $S_2$ be two isolated suborders with $S_1\cap S_2\neq\emptyset$. 
 Then the set $\{\bot_{S_1},\bot_{S_2},\top_{S_1},\top_{S_2}\}$ is a chain.
\end{lem}

\begin{proof} Let us pick an arbitrary $s_{12}\in S_1\cap S_2$.
In the case $\top_{S_1}\in S_2$, the inequality $\top_{S_1}\ord\top_{S_2}$ is easy to see.
If $\top_{S_1}$ is not an element of $S_2$, then we can conclude $\top_{S_2}\ord\top_{S_1}$ from $s_{12}\in S_2$, $s_{12}\ord\top_{S_1}$ and the properties of an isolated suborder.
Symmetrically, we can show that $\bot_{S_1}$ and $\bot_{S_2}$ are comparable.
Now the rest is an easy consequence of $\bot_{S_1},\bot_{S_2}\ord s_{12}\ord\top_{S_1},\top_{S_2}$.
\end{proof}

Now we see that two isolated suborders with a common element can be merged into one isolated suborder:

\begin{lem}\label{lem:unionNotDisjointIso}
 Let $S_1$ and $S_2$ be two isolated suborders with $S_1\cap S_2\neq\emptyset$.
 Then $S_{12}\defeqs S_1\cup S_2$ is an isolated suborder, too.
\end{lem}

\begin{proof} Because $\{\bot_{S_1},\bot_{S_2},\top_{S_1},\top_{S_2}\}$ is a chain due to Lemma~\ref{lem:chainBotTop}, we can assume w.l.o.g. that $\bot_{S_1}\ord\bot_{S_2}$ holds.
If $\top_{S_2}\ord\top_{S_1}$ holds, then we have $[\bot_{S_2},\top_{S_2}]\subseteq[\bot_{S_1},\top_{S_1}]$, and we obtain the claim immediately because isolated suborders are intervals.
We do not have to consider the case $\top_{S_1}\strord\bot_{S_2}$ because of $S_1\cap S_2\neq\emptyset$, so we only have to look at $\bot_{S_1}\ord\bot_{S_2}\ord\top_{S_1}\ord\top_{S_2}$ as last case.

To show that $S_{12}$ is indeed an isolated suborder, we pick arbitrary $s_{12}\in S_{12}$ and $x\notin S_{12}$ such that $s_{12}\ord x$ holds.
In the case $s_{12}\in S_2$ we obtain $x\revord\top_{S_2}=\top_{S_{12}}$ because $S_2$ is an isolated suborder, so let us now assume that $s_{12}\in S_1$ holds.
Due to the properties of $S_1$, we have here $x\revord\top_{S_1}$.
Knowing this,  we can deduce $x\revord\top_{S_2}$ from $\top_{S_1}\in S_2$ and $x\notin S_2$.
The case $x\ord s_{12}$ can be treated by a symmetric argumentation.
\end{proof}

An easy observation now is that $S_1\cup S_2$ is an isolated suborder with bottleneck provided that $S_1$ and $S_2$ are isolated suborders with bottlenecks.
Moreover, if $S_1$ and $S_2$ are nontrivial summit isolated suborders with $\top_{S_1}=\top_{S_2}$ then $S_1\cup S_2$ is also a nontrivial summit isolated suborder with greatest element $\top_{S_1}$ (or, equivalently, greatest element $\top_{S_2}$).
This, together with Lemma~\ref{lem:unionNotDisjointIso}, proves the following theorem:

\begin{thm}\label{the:disBisoUniqueSiso}
 Let $(S,\ord)$ be an ordered set.
 \begin{enumerate}[1.]
  \item Two distinct inclusion-maximal isolated suborders with bottleneck of $(S,\ord)$ are disjoint.
  \item For every $s\in\mymax{S}$, there is at most one nontrivial inclusion-maximal summit isolated suborder with $s$ as greatest element.
 \end{enumerate}
\end{thm}

The next lemma establishes a connection between isolated suborders in an ordered set and a quotient of this order, induced by an isolated suborder.

\begin{lem}\label{lem:isoInQuotient}
 Let $S'$ be an isolated suborder of an ordered set $(S,\ord)$, and let $S_{S'}$ be an isolated suborder of $S\quot\equiv_{S'}$.
 Then $S''\defeqs\bigcup S_{S'}$ is an isolated suborder of $S$.
\end{lem}

\begin{proof} Clearly, $\bot_{S''}=\bot_{S'}$ holds, provided $\bot_{S_{S'}}=S'$, and also $\bot_{S''}=s$, provided $\bot_{S_{S'}}=\{s\}$.
Analogous equalities hold also for $\top_{S''}$ instead of $\bot_{S_{S'}}$.

To show the remaining properties of Definition~\ref{def:isolatedSuborder}, we pick arbitrary $s\in S''$ and $t\notin S''$ with $s\ord t$.
The construction of $S''$ yields both $\equiclass{s}{\equirel_{S'}}\in S_{S'}$ and $\equiclass{t}{\equirel_{S'}}\notin S_{S'}$, and homomorphism properties lead to $\equiclass{s}{\equirel_{S'}}\ord_{\equiv_{S'}}\equiclass{t}{\equirel_{S'}}$.
Because $S_{S'}$ is an isolated suborder, we can therefrom deduce that $\top_{S_{S'}}$ $\ord_{\equiv_{S'}}\equiclass{t}{\equirel_{S'}}$ has to hold.
We note that $\equiclass{t}{\equirel_{S'}}$ and $\top_{S_{S'}}$ are disjoint and consider first the case that $\top_{S_{S'}}$ is a singleton set.
In this case, the equality $\top_{S_{S'}}=\{\top_{S''}\}$ holds which implies $\top_{S''}\ord t$.
Otherwise, i.e., if $\top_{S_{S'}}$ has more than one element, we have $\top_{S_{S'}}=S'$, implying $\top_{S''}=\top_{S'}$.
Here too, we have $\top_{S''}\ord t$ by homomorphism properties and due to disjointness of $\equiclass{t}{\equirel_{S'}}$ and $\top_{S_{S'}}$.
A symmetric argumentation applies to the case $s\revord t$, hence $S''$ is an isolated suborder according to Definition~\ref{def:isolatedSuborder}.
\end{proof}

Next we extend this claim to isolated suborders with bottleneck:

\begin{lem}\label{lem:bisoInQuotient}
 Let $S'$ be an isolated suborder of an ordered set $(S,\ord)$, and let $S_{S'}$ be an isolated suborder with bottleneck of $S\quot\equiv_{S'}$.
 Then $S''\defeqs\bigcup S_{S'}$ is an isolated suborder with bottleneck of $S$.
\end{lem}

\begin{proof} $S''$ is an isolated suborder by Lemma~\ref{lem:isoInQuotient}, and by assumption we can choose an arbitrary bottleneck $B\in S\quot\equiv_{S'}$ of $S_{S'}$.
Now we have the following three possibilities:
\begin{enumerate}[1.]
 \item $\top_{S_{S'}}=S'$: In this case, we have $\top_{S''}=\top_{S'}$ and $B=\{b\}$ for some $b\in S$.
 Hence $b$ is a bottleneck according to Definition~\ref{def:bottleneck} by homomorphism properties because all elements of the (nonempty) interval $]S',B]$ are singleton sets.
 \item $B=S'$: Here $\top_{S_{S'}}=\{\top_{S''}\}$ holds, and our goal is to show that $\bot_{S'}$ is a bottleneck of $S''$.
 By assumption, $[\{\top_{S''}\},S']$ is a chain in $S\quot\equiv_{S'}$, hence $[\top_{S''},\bot_{S'}]$ is a chain in $S$ due to the fact that $S\quot\equiv_{S'}$ consists only of singleton sets except possibly $S'$.
 Now it is easy to check the remaining properties of Definition~\ref{def:bottleneck}.
 \item $\top_{S_{S'}}\neq S'\wedge B\neq S'$: Here, we first investigate the case $S'\in[\top_{S_{S'}},B]$. 
 Under this condition, $B$ is also a bottleneck of $S_{S'}$, and the argumentation can be carried out analogously to the previous case.
 If $S'\notin[\top_{S_{S'}},B]$ holds, $S\quot\equiv_{S'}$ consists of singleton sets only, and the properties of Definition~\ref{def:bottleneck} follow easily from homomorphism properties. \qedhere
\end{enumerate}
\end{proof}

An analogous lemma holds also for summit isolated suborders:

\begin{lem}\label{lem:sisoInQuotient}
 Let $S'$ be an isolated suborder of an ordered set $(S,\ord)$, and let $S_{S'}$ be a summit isolated suborder of $S\quot\equiv_{S'}$.
 Then $S''\defeqs\bigcup S_{S'}$ is a summit isolated suborder of $S$.
\end{lem}

\begin{proof} By Lemma~\ref{lem:isoInQuotient}, we know that $S''$ is an isolated suborder. 
If $\top_{S_{S'}}=\{s\}$ holds for some $s\in S$, then $s$ is a maximal element in $(S,\ord)$ by homomorphism properties.
Moreover, also by homomorphism properties, $s$ is the greatest element of $S''$ so $S''$ is a summit isolated suborder in this case.
Let us now assume that $\top_{S_{S'}}=S'$ holds.
In this case, $S'$ is a maximal element of $(S/\equirel_{S'},\ord_{\equirel_{S'}})$, hence $\top_{S'}$ is a maximal element of $(S,\ord)$.
By construction and homomorphism properties, $\top_{S'}$ is the greatest element of $S''$.
\end{proof}

In our algorithm we may make use of consecutive quotients induced by various isolated suborders.
A (possibly infinite) sequence $S_0,S_1,S_2,\hdots$ of ordered sets is called a \emph{quotient sequence} if for all $i$ the ordered set $S_{i+1}$ can be written as the quotient $S_{i+1}=S_i\quot\equiv_{S'_i}$ for some isolated suborder $S'_i$ of $S_i$.
Consecutive quotient formation leads to ordered sets whose carrier sets have an increasing depth of set nesting.
To be type correct, we introduce the notation $\bigcup^n\setsystem{C}$, defined inductively by $\bigcup^0\setsystem{C}\defeqs\setsystem{C}$ and $\bigcup^{n+1}\setsystem{C}\defeqs\bigcup(\bigcup^n\setsystem{C})$.
Intuitively spoken, this operation removes $n$ set brackets from the elements of $\setsystem{C}$ and joins them all.

In a quotient sequence, an inclusion-maximal summit isolated suborder containing a fixed maximal element can appear as most once as a factor (if we abstract from set parentheses):

\begin{lem}\label{lem:maxSisoQuotient}
 Let $S_0,S_1,S_2,\hdots$ be a quotient sequence, and assume that there exists an index $i$ such that $S_{i+1}=S_i\quot\equiv_{S'_i}$ holds for an inclusion-maximal useful summit isolated suborder $S'_i$ with greatest element $\top_{S'_i}$.
 Then there is no $S_j$ with $j>i$ containing a useful summit isolated suborder $S'_j$ with $\top_{S'_i}\in\bigcup^{j-i}S'_j$.
\end{lem}

\begin{proof} Assume that some $S_j$ with $j>i$ contains a useful summit isolated suborder $S'_j$.
Then we could construct an inclusion-maximal summit isolated suborder $S''_i\supsetneq S'_i$ of $S_i$ from $S'_j$ backwards along the lines of Lemmata~\ref{lem:isoInQuotient} and~\ref{lem:sisoInQuotient}, contradicting the inclusion-maximality of $S'_i$.
\end{proof}

In the next lemma, we show that if we ignore set brackets, an element can appear in at most one inclusion-maximal isolated suborders in a quotient sequence.

\begin{lem}\label{lem:maxBisoQuotient}
 Let $S_0,S_1,S_2,\hdots$ be a quotient sequence such that $S_{i+1}=S_i\quot\equiv_{S'_i}$ holds for an inclusion-maximal useful isolated suborder with bottleneck $S'_i$ for all $i\geq 0$.
 Then $S'_i$ and $\bigcup^{j-i}S'_j$ are disjoint for all $i,j$ with $j>i$.
\end{lem}

\begin{proof} By the first part of Theorem~\ref{the:disBisoUniqueSiso} and Lemma~\ref{lem:bisoInQuotient}, it is obvious that $S'_i$ and $\bigcup S'_{i+1}$ are disjoint for all $i$.
Now the claim follows by straightforward induction.
\end{proof}

\section{Isolated Suborders and Closure Systems}\label{sec:isoAndCloSys}

After investigating isolated suborders and their properties, we now turn our attention to the interplay between isolated suborders and closure systems.
First we show how closure systems on an ordered set give rise to closure systems of a quotient:

\begin{lem}\label{lem:closysInQuotientOrder}
 Let $(S,\ord)$ be an ordered set, $S'$ an isolated suborder of $(S,\ord)$, and consider a closure system $C$ of $(S,\ord)$.
 \begin{enumerate}
  \item If $C\cap S'=\emptyset$, then $\ensemble{C}$ is a closure system of $S\quot\equiv_{S'}$.
  \item If $C\cap S'\neq\emptyset$, then $\ensemble{(C\bez S')}\cup\{S'\}$ is a closure system of $S\quot\equiv_{S'}$.
 \end{enumerate}
\end{lem}

\begin{proof} Because $\equiv_{S'}$ is order generating, $S\quot\equiv_{S'}$ is a homomorphic image of $(S,\ord)$.
Deploying this fact, Definition~\ref{def:closureSystem} can now be verified easily on $S\quot\equiv_{S'}$ for the two cases of the Lemma.
\end{proof}

By the nature of things, homomorphisms work in general only in one direction, so we expect to have a harder task to show how closure systems of a quotient can induce closure systems of the original ordered set.
First we introduce a notion for ``almost'' closure systems:

\begin{defi}\label{def:preclosure}
 Let $(S,\ord)$ be an ordered set with greatest element $\top$. 
 A subset $C\subseteq S$ is called a \emph{preclosure system} of $(S,\ord)$ if $C\cup\{\top\}$ is a closure system of $(S,\ord)$.
 The set of all preclosure systems of $(S,\ord)$ is denoted by $\pcsystem{S}$.
\end{defi}

Clearly, every closure system on an ordered set with greatest element is also a preclosure system, and the empty set is a preclosure system on ordered sets with a greatest element.
An important observation for the algorithm we will develop is that $|\pcsystem{S}|=2\cdot|\csystem{S}|$ holds if $\csystem{S}$ is finite and $S$ has a greatest element (note that this is a precondition to make $\pcsystem{S}$ well-defined).
Another crucial fact in the further course is that a nonempty preclosure system contains a least element majorizing $\bot$ (if the order under consideration has a least element at all):

\begin{lem}\label{lem:preclosureBotDominated}
 Let $C$ be a nonempty preclosure system of an ordered set $(S,\ord)$ with least element $\bot_S$ and greatest element $\top_S$.
 Then there is a least element $c\in C$ majorizing $\bot_S$.
\end{lem}

\begin{proof} In the case $\top_S\in C$, $C$ is even a closure system, and Definition~\ref{def:closureSystem} entails the claim obviously.
Otherwise, we define the closure system $C'$ by $C'\defeqs C\dot{\cup}\{\top_S\}$, and by definition of a closure system there is a least $c'\in C'$ majorizing $\bot_S$.
This element $c'$ can not be $\top_S$ since $C'$ contains at least one element except $\top_S$ (recall the $C$ was supposed to be nonempty), so we conclude $c'\in C$.
\end{proof}

Now we can use preclosure systems to describe the intersection of an isolated suborder and a closure system:

\begin{lem}\label{lem:closureSuborderPreclosure}
 Let $C$ be a closure system on an ordered set $(S,\ord)$, and let $S'$ be an isolated suborder of $S$ with greatest element $\top_{S'}$ and least element $\bot_{S'}$.
 Then $C'\defeqs C\cap S'$ is a preclosure system of $S'$.
 Moreover, if $S'$ is a summit isolated suborder, then $C'$ is a closure system of $S'$.
\end{lem}

\begin{proof} Let $S'$ be an arbitrary isolated suborder and consider an arbitrary $s'\in S'$.
Because $C$ is a closure system, it contains a least $c$ which majorizes $s'$.
If $c\in S'$ holds, then $c$ is by construction also an element of $C'$.
In the other case $c\notin S'$ we have $\top_{S'}\ord c$ by definition of an isolated suborder.
Now it is obvious that $\top_{S'}$ is a smallest element of $C'$ majorizing $s'$.
This shows the first claim, so let us now assume that $S'$ is even a summit isolated suborder. 
Then $\top_{S'}\in\mymax{S}$ holds by definition, hence $C$ has to contain also $\top_{S'}$ wherefrom the second claim follows immediately.
\end{proof}

The next lemma is in some sense a converse of Lemma~\ref{lem:closysInQuotientOrder} in the case of isolated suborders with bottleneck:

\begin{lem}\label{lem:bottleneckISOPreclosys}
 Let $(S,\ord)$ be an ordered set and $S'$ an isolated suborder of $S$ such that $\top_{S'}$ has a least bottleneck $b$.
 Assume that $C_{S'}$ is a preclosure system of $S'$, and let $C'$ be a closure system of $S\quot\equiv_{S'}$ with $S'\in C'$.
 Then $C\defeqs\bigcup(C'\bez\{S'\})\cup C_{S'}$ is a closure system of $(S,\ord)$.
\end{lem}

\begin{proof} Let us pick an arbitrary $s\in S$ with the goal to show that there is a least $c\in C$ with $s\ord c$ in order to fulfill Definition~\ref{def:closureSystem}.
To this end, we have several cases to consider:
\begin{enumerate}[1.]
 \item $s\notin S'$: In this case, we have $[s]_{\equirel_{S'}}=\{s\}$, and by assumption there is a least $c'\in C'$ with $\{s\}\ord_{\equirel_{S'}}c'$. 
 Now we have the following possibilities:
 \begin{enumerate}[a)]
   \item $c'\neq S'$: then we have $c'=\{c''\}$ for some $c''\in S$. 
   By homomorphism properties, $c''$ majorizes $s$ in $(S,\ord)$, so let us pick an arbitrary $\hat{c}\in C$ with $s\ord\hat{c}$.
   If $[\hat{c}]_{\equirel_{S'}}=\{\hat{c}\}$ holds, then we have $c'\ord_{\equirel_{S'}}\{\hat{c}\}$ since $C'$ is a closure system (note that by construction of $C$, $\{\hat{c}\}$ has to be an element of $C'$). 
   Otherwise, we have $[\hat{c}]_{\equirel_{S'}}=S'$, and because $C'$ is a closure system, we have $c'\ord_{\equirel_{S'}}S'$. 
   In both cases we conclude $c''\ord\hat{c}$ by homomorphism properties.
   \item $c'=S'$: here we distinguish the following cases:
   \begin{enumerate}[i)]
     \item $C_{S'}\neq\emptyset$: by Lemma~\ref{lem:preclosureBotDominated}, there is a least $c''\in C_{S'}$ majorizing $\bot_{S'}$. 
     By homomorphism properties and the assumption $\{s\}\ord_{\equirel_{S'}}S'$ we have $s\ord\bot_{S'}$, and by transitivity we get $s\ord c''$. 
     Let us now consider an arbitrary $\hat{c}\in C$ with $s\ord\hat{c}$. 
     If $\hat{c}\in S'$, then we get $\hat{c}\in C_{S'}$ by construction of $C$.
     Because $c''$ is the least element of $C_{S'}$, we obtain $c''\ord\hat{c}$ immediately.
     In the case $\hat{c}\notin S'$ we have $\{\hat{c}\}\in C'$ by construction of $C$, and as above we get $\{s\}\ord_{\equirel_{S'}}\{\hat{c}\}$.
     Because $C'$ is a closure system, we obtain $S'\ord_{\equirel_{S'}}\{\hat{c}\}$ and hence $c''\ord\hat{c}$.
     \item $C_{S'}=\emptyset$: because of $b\notin S'$ we have $[b]_{\equirel_{S'}}=\{b\}$, and since $C'$ is a closure system, there is a least $b'\in C'$ majorizing $\{b\}$.
     Moreover, $b'=\{c''\}$ has to hold for some $c''\in S$ (note that we have $S'\strord_{\equirel_{S'}}\{b\}\ord_{\equirel_{S'}}b'$), and by homomorphism properties and transitivity we obtain $s\ord c''$.
     As usual we pick an arbitrary $\hat{c}\in C$ with $s\ord\hat{c}$ and observe that $[\hat{c}]_{\equirel_{S'}}=\{\hat{c}\}$ holds (note that we assume here $C_{S'}=\emptyset$).
     From $s\ord\hat{c}$ we derive $\{s\}\ord_{\equirel_{S'}}\{\hat{c}\}$ and hence $S'\ord_{\equirel_{S'}}\{\hat{c}\}$ (recall $c'=S'$ and the properties of $c'$).
     However, due to $C_{S'}=\emptyset$ and construction of $C$, we obtain $\hat{c}\notin S'$, and by order theory and homomorphism this leads to $\{b\}\ord_{\equirel_{S'}}\{\hat{c}\}$ from where we conclude that $b'\ord_{\equirel_{S'}}\{\hat{c}\}$ and eventually $c''\ord\hat{c}$ hold.
   \end{enumerate}
 \end{enumerate}
 \item $s\in S'$: clearly, $S'$ is the least element of $C'$ majorizing $[s]_{\equirel_{S'}}$ since $[s]_{\equirel_{S'}}=S'$ and $S'\in C'$ hold. 
 We have two cases:
 \begin{enumerate}[a)]
  \item $C_{S'}$ contains an element $s'$ majorizing $s$: then $C_{S'}$ contains also a minimal element $c''$ majorizing $s$, and let $\hat{c}$ be an arbitrary element of $C$ such that $s\ord\hat{c}$ holds.
  In the case $\hat{c}\in C_{S'}$ we get $c''\ord\hat{c}$ immediately by the choice of $c''$, so let us assume that $\hat{c}\notin S'$ holds.
  However, here we have $\top_{S'}\ord\hat{c}$ by properties of an isolated suborder, implying $c''\ord\hat{c}$.
  \item $C_{S'}$ contains no element majorizing $s$: as above, there is a least $b'=\{c''\}\in C'$ majorizing $\{b\}$.
  Clearly, $c''$ majorizes $s$ in $(S,\ord)$ so we pick an arbitrary $\hat{c}\in C$ with $s\ord\hat{c}$.
  By properties of an isolated suborder we have $\top_{S'}\ord\hat{c}$, and because $C_{S'}$ does not contain $\top_{S'}$ (otherwise $\top_{S'}$ would be an element of $C_{S'}$ majorizing $s$), we can deduce even $b\ord\hat{c}$.
  Now it is clear that $c''\ord\hat{c}$ holds due to $\{b\}\ord_{\equirel_{S'}}\{\hat{c}\}$ and closure properties of $C'$.
 \end{enumerate}
\end{enumerate}
In all cases, we constructed in the form of $c''$ a least element of $C$ majorizing $s$.
\end{proof}

\begin{rem} It is necessary for the correctness of Lemma~\ref{lem:bottleneckISOPreclosys} to require that $\top_{S'}$ has a bottleneck.
This can be seen in Figure~\ref{fig:ISOWithoutBottleneck}.
At the left, a preclosure system, indicated by encircled elements, on an isolated suborder without bottleneck of its top element, indicated by an ellipse, is shown.
In the middle of the figure, we see a closure system on the associated quotient order, indicated by circles.
However, executing the construction from Lemma~\ref{lem:bottleneckISOPreclosys} leads to the set of encircled elements in the right picture which does not contain a least element majorizing the middle element (the only unencircled element), and hence is no closure system. 
This shows one effect of a bottleneck: it releases the top element of an isolated suborder from the responsibility of being the least element being majorized by two elements above it.
Moreover, it it necessary to require that $\top_{S'}$ has even a least bottleneck: consider the ordered set $S=([0,1],\ord)$ (where $\ord$ denotes the usual order on the reals) and the isolated suborder $S'=(\{0\},\ord)$. 
We choose the empty set as preclosure $C_{S'}$ of $S'$ and $\ensemble{[0,1]}$ as closure system $C'$.
Then the construction from above yields for $C$ the set $]0,1]$ which is no closure system since it contains no least element majorizing $0$.
\end{rem}

\begin{figure}
 \includegraphics[width=1.1\textwidth]{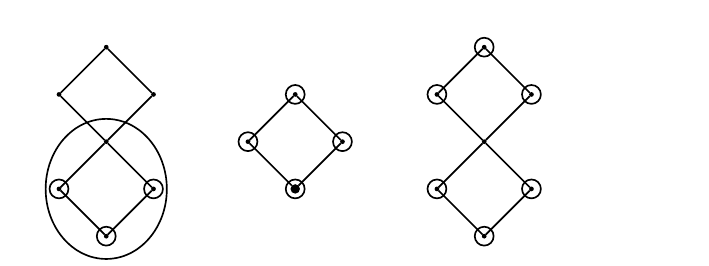}
 \caption{A preclosure system on an isolated suborder (left), a closure system on a quotient (middle) and no closure system on the original order (right)}
\label{fig:ISOWithoutBottleneck}
\end{figure}

The next lemma is a variant of the previous lemma for summit isolated suborders:

\begin{lem}\label{lem:summitISOClosys}
 Let $(S,\ord)$ be an ordered set, and let $S'$ be a summit isolated suborder of $(S,\ord)$.
 Assume that $C_{S'}$ is a closure system of $S'$, and let $C'$ be a closure system of $S/\equirel_{S'}$.
 Then $C\defeqs\bigcup(C'\bez\{S'\})\cup C_{S'}$ is a closure system of $(S,\ord)$.
\end{lem}

\begin{proof} As in the proof of Lemma~\ref{lem:bottleneckISOPreclosys}, we pick an arbitrary $s\in S$ and show the existence of a least element of $C$ majorizing $s$.
We distinguish the following cases:
\begin{enumerate}[1.]
 \item $s\notin S'$: analogously to Lemma~\ref{lem:bottleneckISOPreclosys}, we have $[s]_{\equirel_{S'}}=\{s\}$, and by assumption there is a least $c'\in C'$ with $\{s\}\ord_{\equirel_{S'}}c'$. 
 Now we have the following possibilities:
 \begin{enumerate}[a)]
 \item $c'\neq S'$: this case can be handled exactly like the same case in the proof of Lemma~\ref{lem:bottleneckISOPreclosys}.
 \item $c'=S'$: here we have the following cases:
 \begin{enumerate}[i)]
 \item $C_{S'}\neq\emptyset$: this case can also be handled along the lines of Lemma~\ref{lem:bottleneckISOPreclosys} (note that $C_{S'}$ as closure system is also a preclosure system).
 \item $C_{S'}=\emptyset$: in this case we can not resort to the proof of Lemma~\ref{lem:bottleneckISOPreclosys} because the argument there uses properties of a smallest bottleneck.
 However, $C_{S'}$ can not be empty because it has to contain the greatest element of $S'$.
 \end{enumerate}
 \end{enumerate}
 \item $s\in S'$: by properties of a closure system, $C_{S'}$ contains a least element $c'$ majorizing $s$.
 However, By construction of $C$ and because $S'$ is a summit isolated suborder, all elements of $C\bez C_{S'}$ majorizing $s$ are contained in $C_{S'}$
\end{enumerate}
In all cases, we deduced the existence of a least element of $C$ majorizing $s$.
\end{proof}

Finally, we consider the case that an isolated suborder with bottleneck does not appear in a closure system of the quotient:

\begin{lem}\label{lem:closysWithoutISO}
 Let $(S,\ord)$ be an ordered set, and let $S'$ be an isolated suborder with bottleneck of $(S,\ord)$.
 Assume that $C'$ is a closure system on $S/\hspace*{-1mm}\equirel_{S'}$ such that $S'\notin C'$.
 Then $C\defeqs\bigcup C'$ is a closure system on $S$.
\end{lem}

\begin{proof} Routinely, we pick an arbitrary $s\in S$ with the obligation to show the existence of a least element of $C$ majorizing $s$.
Then the case $s\notin S'$ is analogous to case 1.a) of the proof of Lemma~\ref{lem:bottleneckISOPreclosys} (note that $S'\notin C'$ was assumed).
Similarly, the case $s\in S'$ can be handled along the lines of case 2.b) of the proof of Lemma~\ref{lem:bottleneckISOPreclosys}.
\end{proof}

Now we are ready to state our main results about the relationships between closure systems on ordered sets and quotients thereof in the following two theorems:

\begin{thm}\label{the:closuresBISO}
 Let $S'$ be an isolated suborder with bottleneck of an ordered set $(S,\ord)$, and consider a subset $C\subseteq S$.
 \begin{enumerate}
  \item Assume that $C'\defeqs C\cap S'\neq\emptyset$ holds. Then $C$ is a closure system of $S$ iff $C'$ is a nonempty preclosure system of $S'$ and $\ensemble{(C\bez S')}\cup\{S'\}$ is a closure system of $S\quot\equiv_{S'}$.
  \item Assume that $C\cap S'=\emptyset$ holds. Then $C$ is a closure system of $S$ iff $\ensemble{C}$ is a closure system of $S\quot\equiv_{S'}$.
 \end{enumerate}
\end{thm}

\begin{proof} This follows now easily from Lemmata~\ref{lem:closysInQuotientOrder},~\ref{lem:closureSuborderPreclosure},~\ref{lem:bottleneckISOPreclosys} and~\ref{lem:closysWithoutISO}\end{proof}

\begin{thm}\label{the:closuresSISO}
 Let $S'$ be a summit isolated suborder of an ordered set $(S,\ord)$, and consider a subset $C\subseteq S$.
 Then $C$ is a closure system of $S$ iff $C\cap S'$ is a closure system of $S'$ and $\ensemble{(C\bez S')}\cup\{S'\}$ is a closure system of $S\quot\equiv_{S'}$.
\end{thm}

\begin{proof} This follows simply from Lemmata~\ref{lem:closysInQuotientOrder},~\ref{lem:closureSuborderPreclosure} and~\ref{lem:summitISOClosys}.\end{proof}

\section{Computing Isolated Suborders}\label{sec:compISO}

The previous results concerned isolated suborders and their relation to closure systems.
In this section, we will become more concrete and deal with the computation of isolated suborders.
As representation for an ordered set, we assume that an ordered set is given by its Hasse diagram as a directed graph $G=(S,E)$ where the edges point ``upwards'', i.e., $(s,t)\in E$ implies $s\strord t$.
Given a path $p=s_1s_2\hdots s_n$, we write $s\in p$ if $s=s_i$ holds for some $i\in[1,n]$ and say that $p$ \emph{contains} $s$.
We call a node $u$ an $(s,t)$\emph{-separator} if every path from $s$ to $t$ contains $u$.
Routinely, we add an element $\bot$ to $S$ and consider the graph $G_{\bot}\defeqs(S_{\bot},E_{\bot})$ with $S_{\bot}\defeqs S\dot{\cup}\{\bot\}$ and $E_{\bot}\defeqs E\dot{\cup}\{(\bot,m)\st m\in\mymin{S}\}$.
Intuitively, this adds a least element to the ordered set.
Additionally, we may add yet another element $\top$ to $S_{\bot}$ to obtain the set $S_{\bot,\top}\defeqs S_{\bot}\dot\cup\{\top\}$ and define the graph $G_{\bot,\top}\defeqs(S_{\bot,\top},E_{\bot,\top})$ by $E_{\bot,\top}\defeqs E_{\bot}\dot{\cup}\{(m,\top)\st m\in\mymax{S}\}$.
Analogously to the construction of $G_{\bot}$, this adds a greatest element to the ordered set under consideration.
All algorithms considered in this section have polynomial running time so these constructions do not affect their asymptotic complexity.

\subsection{Isolated Suborders and Separators}

In the sequel, we will establish a connection between isolated suborders and separators in finite graphs.
In the following two lemmata which form the basis for Theorem~\ref{thm:isoSeparator}, finiteness is a crucial point in order to argue about finite paths.

\begin{lem}\label{lem:isoImplSep}
Let $(S,\ord)$ be a finite ordered set with Hasse diagram $G=(S,E)$, and let $[s_1,s_2]$ be an isolated suborder of $(S,\ord)$.
Then $s_1$ is a $(\bot,s_2)$-separator in $G_{\bot,\top}$, and $s_2$ is an $(s_1,\top)$-separator in $G_{\bot,\top}$.
\end{lem}

\begin{proof}
We first show the claim that $s_1$ is a $(\bot,s_2)$-separator.
To this end, let $p=p_1p_2\hdots p_n$ with $p_1=\bot$ and $p_n=s_2$ be a path in $G_{\bot,\top}$ from $\bot$ to $s_2$.
Then there is an index $i\in[1,n-1]$ with $p_i\notin[s_1,s_2]$ and $p_{i+1}\in[s_1,s_2]$, and we claim that $p_{i+1}=s_1$ holds.
Let us therefore assume that $p_{i+1}\neq s_1$ holds.

First we rule out the case $i=1$: if $p_2\in[s_1,s_2]$ and $p_2\neq s_1$ hold, then we have $s_1\strord p_2$.
However, this contradicts minimality of $p_2$ in $(S,\ord)$ which follows from the construction of $G_{\bot,\top}$.
The crucial point is now that all elements of $p$ under consideration are vertices of $G$ (or, equivalently, elements of $S$).

Now we have $s_1\strord p_{i+1}$ due to $p_{i+1}\in[s_1,s_2]$ and $p_{i+1}\neq s_1$.
Furthermore, we have $p_i\strord p_{i+1}$ by construction of $p$, and now $p_{i}\notin[s_1,s_2]$, $p_{i+1}\in[s_1,s_2]$, and the properties of an isolated suborder imply $p_i\strord s_1$.
Altogether, this means $p_i\strord s_1\strord p_{i+1}$, contradicting the fact that $(p_i,p_{i+1})$ is an edge of a Hasse diagram.
\end{proof}

\begin{lem}\label{lem:sepImpliso}
Let $(S,\ord)$ be a finite ordered set with Hasse diagram $G=(S,E)$, and let $[s_1,s_2]$ be a nonempty interval.
Assume that both $s_1$ is a $(\bot,s_2)$-separator in $G_{\bot,\top}$ and $s_2$ is an $(s_1,\top)$-separator in $G_{\bot,\top}$.
Then $[s_1,s_2]$ is an isolated suborder.
\end{lem}

\begin{proof}
We only show that for all $s'\notin[s_1,s_2]$ and $s''\in[s_1,s_2]$ with $s'\ord s''$ (note that this can be strengthened $s'\strord s''$) the inequality $s'\ord s_1$ holds.
By construction, there is a path $p'=p_1'p_2'\hdots p_{n'}'$ in $G_{\bot,\top}$ with $p_1'=\bot$ and $p_{n'}'=s'$.
Note that $s_1$ can not be an element of $p'$: assume there is an index $i'$ with $p_i'=s_1$.
By construction, $i'<n'$ has to hold which implies $s_1\ord s'$.
Together with $s'\ord s''$ and $s''\in[s_1,s_2]$, this implies $s'\in[s_1,s_2]$, contradicting the choice of $s'$.

Next, there is a path $p''=p_1''p_2''\hdots p_{n''}''$ in $G_{\bot,\top}$ with $p_1''=s'$ and $p_{n''}''=s''$ (this is due to the fact $s'\strord s''$).
Finally, due to $s''\in[s_1,s_2]$, there is a path $\hat{p}=\hat{p}_1\hat{p}_2\hdots\hat{p}_{\hat{n}}$ in $G_{\bot,\top}$ with $\hat{p}_1=s''$ and $\hat{p}_{\hat{n}}=s_2$.
The concatenation of these three paths yields a path $p=p_1p_2\hdots p_n$ which has to contain $s_1$ because $s_1$ is a $(\bot,s_2)$-separator.
However, the discussion about $p'$ shows that $s'$ appears in $p$ before $s_1$, hence we have $s'\ord s_1$.

Clearly, $[s_1,s_2]$ has a least and a greatest element, and the remaining property of Definition~\ref{def:isolatedSuborder} can be shown symmetrically to above.
\end{proof}

From the previous two lemmata, the next theorem follows immediately:

\begin{thm}\label{thm:isoSeparator}
Let $(S,\ord)$ be a finite ordered set with Hasse diagram $G=(S,E)$. 
Then, an interval $[s_1,s_2]$ is an isolated suborder iff $[s_1,s_2]$ is nonempty, $s_1$ is a $(\bot,s_2)$-separator in $G_{\bot,\top}$, and $s_2$ is an $(s_1,\top)$-separator in $G_{\bot,\top}$.
\end{thm}

\subsection{Computing Isolated Suborders}\label{ssec:compISO}

The result from Theorem~\ref{thm:isoSeparator} can be used to determine all useful isolated sublattices with bottleneck in polynomial time.
To this end, we note that it is possible to decide in $\bigo(|V|+|E|)$ time whether a node $s$ is a $(v_1,v_2)$-separator in a graph $G=(V,E)$: we simply remove $s$ and its adjacent edges from $G$ and test whether in the resulting graph $v_2$ is reachable from $v_1$.
Algorithm~\ref{alg:compIsoBottleneck} shows now how to compute useful isolated suborders with bottleneck.

\begin{algorithm}[tbh]
\caption{Computing Maximal Useful Isolated Suborders with Bottleneck}\label{alg:compIsoBottleneck}
\begin{algorithmic}[1]
\Function{Useful Maximal Isolated Suborders with Bottleneck}{Hasse diagram $G=(V,E)$}
\State $G_t\gets$ transitive hull of $G$ 
\State \emph{elemsWithBottleneck} $\gets$ topologically ordered list of all nodes of $G$ with outdegree 1
\State \emph{candidates} $\gets$ empty list of type $V\times V$
\State \emph{maximalIsos} $\gets$ empty list of type $V\times V$
\ForAll {$v\in V$}
\ForAll {$b\in$ \emph{elemsWithBottleneck} with $v\neq b$}
\If{$(v,b)$ is an edge in $G_t$}
\If{$v$ is a $(\bot,b)$-separator in $G_{\bot,\top}$ and $b$ is a $(v,\top)$-separator in $G_{\bot,\top}$}
\State append $(v,b)$ to \emph{candidates}
\EndIf
\EndIf
\EndFor
\EndFor
\ForAll {$v\in V$}
\State append the last occurence of the form $(v,b)$ in \emph{candidates} to \emph{maximalIsos}
\EndFor
\State \Return \emph{maximalIsos}
\EndFunction
\end{algorithmic}
\end{algorithm}

First, we compute the order associated with $G$ which can be done in $\bigo(|V^3|)$ time using the Floyd-Warshall algorithm.
An algorithm running in $\bigo(|V|^{2.371552})$ is given in~\cite{AlmanMatrix}, however, this algorithm is not useable in practice.
Next, we collect all possible top elements of an isolated suborder with bottleneck in the list \emph{elemsWithBottleneck} in topological order.
Using standard methods, this can be carried out in $\bigo(|V|+|E|)$ time.
Subsequently, we initialize two lists of type $V\times V$ to collect intermediate results and for the final return result.

The centerpiece of the algorithm are lines 6 to 14: here we loop over all tuples $(v,b)$ such that $[v,b]$ is an interval with at least two elements, and test whether such an interval fulfills the conditions of Theorem~\ref{thm:isoSeparator}.
If so, then the respective interval is added to the \emph{candidates} list.
There are at most $\bigo(|V|^2)$ tuples of the form $(v,b)$, and the computation in lines 9 and 10 takes $\bigo(|V|+|E|)$ time, so lines 6 to 14 can be carried out in $\bigo(|V|^2(|V|+|E|))$ time.
We assume that in line 7 the elements of \emph{bottlenecks} are processed in their topological order so line 16 ensures that only inclusion-maximal suborders are added to the returned \emph{maximalIsos} list.
Here, the loop from line 15 to line 17 is executed $|V|$ times, and since the \emph{candidates} list can contain $\bigo(|V|^2)$ elements, the running time of this loop is bounded by $\bigo(|V|^3)$.
Altogether, the dominating part is given by lines 6 to 14, so the overall running time is in $\bigo(|V|^2(|V|+|E|))$.

This algorithm is rather simple and can be optimized in many ways.
For example, once we find an isolated suborder $[v,b]$ with bottleneck in lines 9 and 10, we have to look in line 7 only at elements $b'$ with $b\strord b'$ and can skip elements $v'$ with $v\strord v'$ in line 6 due to Lemma~\ref{lem:chainBotTop}.
Moreover, ideas from~\cite{ConteCMP20} or~\cite{ItalianoArticulationPoints12} can be used for the computation of separators and to replace the rather primitive computation in line 9.
Note also that this algorithm returns all useful maximal isolated suborders with bottleneck and hence has a priori a running time in $\Omega(|V|^2)$.
The restriction to finding only one arbitrary useful isolated suborder with bottleneck could also lead to a better running time.

Note that this algorithm can easily be modified for the computation of useful isolated summit suborders: instead of \emph{elemsWithBottleneck}, we simply use a list \emph{maxElements} containing all maximal elements of $G$.
Clearly, this list can be computed in $\bigo(|V|)$ time because we put simply all nodes with outdegree 0 into it.
A topological ordering is not necessary; however, the running time is again dominated by lines 6 to 14, and hence amounts to $\bigo(|V|^2(|V|+|E|))$.

\section{Counting Closure Operators}\label{sec:counting}

In this section we will apply our previous results to counting closure systems.
First, we consider some special cases in Subsection~\ref{ssec:specialCases} and give afterward a recursive algorithm based on isolated suborders in Subsection~\ref{ssec:CountingUsingISOs}.

\subsection{Special Cases}\label{ssec:specialCases}

General orders may consist of several distinct connected components; however, we will see that for counting purposes it suffices to concentrate on orders consisting of one connected component.
The following lemma shows a splitting property for closure systems on ordered sets with two or more connected components.

\begin{lem}\label{lem:disOrderClosures}
 Let $(S_1,\ord_1)$ and $(S_2,\ord_2)$ be ordered sets with $S_1\cap S_2=\emptyset$, and let $C_1\subseteq S_1$ and $C_2\subseteq S_2$ be closures systems of $S_1$ and $S_2$, resp.
 Then $C_1\cup C_2$ is a closure system of $(S_{12},\ord_{12})\defeqs(S_1\cup S_2,\ord_1\cup\ord_2)$.
\end{lem}

\begin{proof} Let $s\in S_{12}$ be arbitrary, and assume w.l.o.g. that $s\in S_1$ holds.
Then there is a least (in $C_1$) element $c\in C_1$ majorizing $s$ (wrt. $\ord_1$).
However, elements from $S_1$ and $S_2$ (and hence from $C_1$ and $C_2$) are incomparable wrt. $\ord_{12}$ by construction, so $c$ is also a least element (wrt. $\ord_{12}$) majorizing $s$ (wrt. $\ord_{12}$).
\end{proof}

This observation entails the following corollary:

\begin{cor}\label{cor:disjointCounting}
 Let $(S,\ord)$ be an ordered set such that there is a partition of $S=S_1\dot{\cup}S_2$ into disjoint nonempty finite subsets $S_1$ and $S_2$ such that for all $s_1\in S_1$ and $s_2\in S_2$ the elements $s_1$ and $s_2$ are incomparable. 
 Then the equality $|\csystem{S,\ord}|=|\csystem{S_1,\ord|_{S_1}}|\cdot|\csystem{S_2,\ord|_{S_2}}|$ holds.
\end{cor}

Of course, the results of Lemma~\ref{lem:disOrderClosures} and Corollary~\ref{cor:disjointCounting} can be extended to ordered sets with more than two connected components by means of induction.

Next, we will derive closed formulae for the number of closure system on some special kinds of ordered sets which may serve as termination case in the algorithm from the following Subsection~\ref{ssec:CountingUsingISOs}. 
However, we are not only interested in the overall number of closure systems but also in the number of closure systems containing a given subset of the ordered set under consideration.
To deal with this situation, we introduce the notations $\csystem{S}_T\defeqs\{C\in\csystem{S}\st T\subseteq C\}$, $\csystem{S}^{-S}_T\defeqs\{C\in\csystem{S}_T\st S\cap C=\emptyset\}$, and $\csystem{S}^S_T\defeqs\{C\in\csystem{S}_T\st S\cap C\neq\emptyset\}$ for subsets $T\subseteq S$ of an ordered set $(S,\ord)$.
With this notation, we have the trivial equalities $\csystem{S}=\csystem{S}_{\emptyset}$ (the empty set imposes no constraints) and $\csystem{S}_{T}=\csystem{S}_{T\bez\{m\}}=\csystem{S}_{T\cup\{m\}}$ for each $m\in\mymax{S}$ (every closure system has to contain every maximal element).
Moreover, since $\csystem{S}_T$ is the disjoint union of $\csystem{S}_{T\cup\{s\}}$ and $\csystem{S}^{-\{s\}}_T$, we have the equality $|\csystem{S}_T|=|\csystem{S}_{T\cup\{s\}}|+|\csystem{S}^{-\{s\}}_T|$. In particular, this means that we do not need explicit formulae for $\csystem{S}^{-\{s\}}_T$.

The first special case is that of a chain:

\begin{lem}\label{lem:closureNumberChain}
 Let $(S,\ord)$ be a chain with $n$ elements and consider an arbitrary $T\subseteq S$.
 Then we have $|\csystem{S}_T|=2^{n-1-|T\bez\{\top_S\}|}$.
\end{lem}

\begin{proof} It is straightforward to see that for a finite chain $(S,\ord)$ a set $C\subseteq S$ is a closure system according to Definition~\ref{def:closureSystem} iff it contains $\top_S$.
Now the claim follows now the formula for the cardinality of power sets.
\end{proof}

Next, we consider orders with only one layer of elements between the bottom and top element:

\begin{defi}\label{def:diamond}
 An ordered set $(S,\ord)$ is called a \emph{diamond of width $n$} if its carrier set $S=\{\bot_S,\top_S,b_1,\dots,b_n\}$ consists of $n+2$ pairwise different elements 
 and $b_i\incomparable b_j$ holds for all $i\neq j$.
 The elements $(b_i)_{1\leq i\leq n}$ are called the \emph{belt elements} of $(S,\ord)$.
\end{defi}

\begin{lem}\label{lem:closureNumberDiamond}
 Let $(S,\ord)$ be a diamond of width $n$ and let $B$ be the set of its belt elements. 
 Then the following holds:
 \begin{enumerate}[1.]
  \item $\bot_S\in T\Rightarrow|\csystem{S}_T|=2^{n-|T\bez\{\top_S\}|}$
  \item $\bot_S\notin T\wedge |T\cap B|>1\Rightarrow|\csystem{S}_T|=2^{n-|T\bez\{\top\}|}$
  \item $\bot_S\notin T\wedge T\cap B=\{b_i\}\Rightarrow|\csystem{S}_T|=2^{n-1}+1$
  \item $\bot_S\notin T\wedge T\cap B=\emptyset\Rightarrow|\csystem{S}_T|=2^n+n+1$
 \end{enumerate}
\end{lem}

\begin{proof} 1. Here, all elements of $\csystem{S}_T$ have the form $\{\bot_S,\top_S\}\cup T\cup B'$ for some $B'\subseteq B\bez T$.
However, $T$ occupies already $|T\bez\{\top_S\}|$ places in $B$ so the claim follows again from the cardinality formula for power sets.

2. Because closures have to contain a least element majorizing every element, $|T\cap B|>1$ implies $\bot\in C$ for every $C\in\csystem{S}_T$ which reduces this case to the previous one.

3. Consider a closures system $C\in\csystem{S}_T$. 
If $\bot\notin C$ holds then $b_j\notin C$ has to hold for all $b_i\neq b_j\in B$ since $C$ has to contain a least element majorizing any element. Hence, the only possibility in this case is $C=\{\bot_S,b_i,\top_S\}$.
The case $\bot\in C$ can be treated analogously to the first case and the result follows from summing up.

4. We have $2^n$ closure systems of the form $\{\bot_S,\top_S\}\cup B'$ with $B'\subseteq B$, $n$ of the form $\{b_i,\top_S\}$ and the trivial closure system $\{\top_S\}$.
\end{proof}

A concept similar to diamonds are bottomless diamonds:

\begin{defi}\label{def: bottomlessDiamond}
An ordered set $(S,\ord)$ is called a \emph{bottomless diamond of width $n$} if its carrier set $S=\{\top_S,b_1,\dots,b_n\}$ consists of $n+1$ pairwise different elements and $b_i\incomparable b_j$ holds for all $i\neq j$.
The elements $(b_i)_{1\leq i\leq n}$ are called the \emph{belt elements} of $(S,\ord)$.
\end{defi}

\begin{lem}\label{lem:closureNumberBottomlessDiamond}
Let $(S,\ord)$ be a bottomless diamond of width $n$ and let $B$ be the set of its belt elements.
Then $|\csystem{S}_T|=2^{n-|T\bez\{\top_S\}|+1}$.
\end{lem}

\begin{proof} Clearly, every $C\subseteq S$ is a closure system iff $\top_S\in C$ holds so the claim follows by elementary combinatorics.\end{proof}

\subsection{Counting Closures using Isolated Suborders}\label{ssec:CountingUsingISOs}
From now on we assume that every ordered set under consideration is finite because we are interested in developing an algorithm for counting the number of closure systems.
The algorithm we are going to introduce in Subsection~\ref{ssec:CountingUsingISOs} will make recursive calls computing the number of (pre)closure systems containing already processed elements.

Let $(S,\ord)$ be an ordered set and consider an isolated suborder with bottleneck $S'\subseteq S$ of $S$ and a subset $T\subseteq S$ with $T\cap S'=\emptyset$.
Then we can partition the set $\csystem{S}_T$ into the two disjoint sets $\csystem{S}_T^{S'}$ and $\csystem{S}_T^{-S'}$ (where the first one consists of all elements from $\csystem{S}_T$ containing an element from $S'$ and the second one consists of all elements from $\csystem{S}_T$ containing no element from $S'$).
By the first part of Theorem~\ref{the:closuresBISO}, we obtain the following equation (we have to subtract $1$ in order not to count the empty preclosure twice):
\begin{equation}\label{equ:countBISOIncl}
|\csystem{S}_T^{S'}|=|\csystem{S\quot\equiv_{S'}}_{\ensemble{T}\cup\{S'\}}|\cdot(|\pcsystem{S'}|-1).
\end{equation}

By an analogous argumentation we obtain the following equation by means of the second part of Theorem~\ref{the:closuresBISO}:
\begin{equation}\label{equ:countBISOExcl}
|\csystem{S}_T^{-S'}|=|\csystem{S\quot\equiv_{S'}}_{\ensemble{T}}^{-\{S'\}}|.
\end{equation}


Now we can argue as follows:
\begin{align*}  
|\csystem{S}_T|
&=|\csystem{S}_T^{S'}|+|\csystem{S}_T^{-S'}|
\tag{$\csystem{S}_T=\csystem{S}_T^{S'}\dot\cup\csystem{S}_T^{-S'}$}\\[2pt]
&=|\csystem{S\quot\equiv_{S'}}_{\ensemble{T}\cup\{S'\}}|\cdot(|\pcsystem{S'}|-1)+|\csystem{S}_T^{-S'}|
\tag{Equation~\ref{equ:countBISOIncl}}\\[2pt]
&=|\csystem{S\quot\equiv_{S'}}_{\ensemble{T}\cup\{S'\}}|\cdot(2\cdot|\csystem{S'}|-1)+|\csystem{S}_T^{-S'}|
\tag{$|\pcsystem{S'}|=2\cdot|\csystem{S'}|$}\\[2pt]
&=|\csystem{S\quot\equiv_{S'}}_{\ensemble{T}\cup\{S'\}}|\cdot(2\cdot|\csystem{S'}|-1)+|\csystem{S\quot\equiv_{S'}}_{\ensemble{T}}|-|\csystem{S\quot\equiv_{S'}}_{\ensemble{T}\cup\{S'\}}|
\tag{$|\csystem{S\quot\equiv_{S'}}_{\ensemble{T}}|=|\csystem{S\quot\equiv_{S'}}_{\ensemble{T}\cup\{S'\}}|+|\csystem{S\quot\equiv_{S'}}^{-\{S'\}}_{\ensemble{T}}|$}\\[2pt]
&=|\csystem{S\quot\equiv_{S'}}_{\ensemble{T}\cup\{S'\}}|\cdot2(|\csystem{S'}|-1)+|\csystem{S\quot\equiv_{S'}}_{\ensemble{T}}|
\tag{elementary calculus}
\end{align*}

Altogether, we obtain the following equation:
\begin{equation}\label{equ:count_quot_bil}
 |\csystem{S}_{T}|=|\csystem{S\quot\equiv_{S'}}_{\ensemble{T}\cup\{S'\}}|\cdot2(|\csystem{S'}|-1)+|\csystem{S\quot\equiv_{S'}}_{\ensemble{T}}|
\end{equation}

Analogous considerations for a summit isolated suborder $S'$ guide us to the following formula by means of Theorem~\ref{the:closuresSISO}:
\begin{equation}\label{equ:count_quot_sil}
 |\csystem{S}_{T}|=|\csystem{S\quot\equiv_{S'}}_{\ensemble{T}}|\cdot|\csystem{S'}|
\end{equation}

Clearly, $|\csystem{S}|=|\csystem{S}_{\emptyset}|$ holds, so we can use Equations~(\ref{equ:count_quot_bil}) and~(\ref{equ:count_quot_sil}) for a recursive algorithm if the ordered set under consideration contains a useful summit isolated suborder or a useful isolated suborder with bottleneck.
However, this is only feasible if the isolated suborder $S'$ and the set $T$ are disjoint.
Fortunately, we can ensure this due to Lemmata~\ref{lem:maxSisoQuotient} and~\ref{lem:maxBisoQuotient} by choosing first - if possible - an inclusion-maximal nontrivial summit isolated suborder followed by the choice of inclusion maximal isolated suborders with bottleneck.
A formal description of this idea is given in Algorithm~\ref{alg:counting}.
Of course, if the ordered set under consideration does not contain some kind of useful isolated suborder or does not have a special structure for which a closed formula is available, then we have to resort to some kind of brute force.

Let us now take a look at the functioning of this algorithm.
First (lines 2 to 4), it is tested whether the ordered set under consideration is of a special structure.
For the examples given here, it is easy to see that this can be done in linear time.
In every recursive call of \texttt{\#CLOSURES}, the cardinality of the ordered sets in the first arguments is strictly smaller than the cardinality of the ordered set from the first argument of the function call.
The first reason for this is that every isolated suborder with bottleneck or every nontrivial summit isolated suborder $S'$ is a strict subset of $S$.
Moreover, we consider only useful isolated suborders $S'$, hence $S\quot\equiv_{S'}$ contains strictly less element than $S$.
This enforces termination in the sense that either an ordered set with a special structure (for which a closed formula is known) is obtained or some brute force method is applied.

\begin{algorithm}[tbh]
\caption{Counting Closure Operators}\label{alg:counting}
\begin{algorithmic}[1]
\Function{\#closures}{ordered set $S$, set $T$}
\If{some special case from Subsection~\ref{ssec:specialCases} is applicable}
\State \Return the respective number
\EndIf
\If{$S$ has a useful summit isolated suborder}
\State $S'\gets$ an inclusion maximal useful summit suborder
\State \Return \Call{\#closures}{$S\quot\equiv_{S'}$,$\ensemble{T}$}$\cdot$\Call{\#closures}{$S'$, $\emptyset$}
\EndIf
\If{$S$ has a useful isolated suborder with bottleneck}
\State $S'\gets$ an inclusion maximal useful isolated suborder with bottleneck
\State \Return \Call{\#closures} {$S\quot\equiv_{S'}$, $\ensemble{T}\cup\{\ensemble{S'}\}$} $\cdot$ 2(\Call{\#closures} {$S'$, $\emptyset$}-1)+ \\ \hspace*{3cm}\Call{\#closures}{$S\quot\equiv_{S'}$,${\ensemble{T}}$}
\EndIf
\State compute and return $|\csystem{S}_{T}|$ by some brute force algorithm
\EndFunction
\end{algorithmic}
\end{algorithm}

As described in Algorithm~\ref{alg:compIsoBottleneck} and the following discussion, isolated suborders of interest can be computed in polynomial time.
Hence, we can assume that the test whether an order has a certain structure (lines 2 to 4) and the computation of isolated sublattices in the algorithm are bounded by a polynomial $p$. 
Let us now assume that a brute force algorithm takes $c^{|S|}$ time for some $c>1$ to compute the number of all closure systems.
Furthermore, let us consider a sequence $(S_1,\ord_1)$, $(S_2,\ord_2)$, $\hdots$ of finite ordered sets with $|S_i|<|S_{i+1}|$ such that no $S_i$ has a useful isolated suborder (e.g., choose for $(S_i,\ord_i)$ the power set lattice of $[1,\hdots n]$).
Furthermore, let $(S'_i,\ord'_i)$ be an isomorphic copy of $(S_i,\ord_i)$, and define the ordered set $(T_i,\ord_{T_i})$ by $T_i\defeqs S_i\dot\cup S_{i+1}$ and $\ord_{T_i}\defeqs\ord_i\dot\cup\ord'_i\dot\cup S_i\times S'_i$.
Then $T_i$ contains exactly one useful summit isolated suborder (namely $S'_i$).
An immediate application of a brute force algorithm to $T_i$ would now take $\bigo{(|T_i|^c)}$ time.
However, Algorithm~\ref{alg:counting} will make two recursive calls (note that the condition in line 5 is true) on instances of sizes $\frac{|T_i|}{2}$ and $\frac{|T_i|}{2}+1$.
Hence, the running time in this case will amount to $\bigo(p(|T_i|)+p(\frac{|T_i|}{2})+c^{\frac{|T_i|}{2}}+p(\frac{|T_i|}{2}+1)+c^{\frac{|T_i|}{2}+1})\in$ $\bigo{(c^{\frac{|T_i|}{2}+1})}$.
Asymptotically, this running time is dominated strictly by $\bigo(T_i^c)$, so Algorithm~\ref{alg:compIsoBottleneck} will lead to a speed up.

Another argument in favor of Algorithm~\ref{alg:counting} can be given by means of the master theorem.
Assume that we have an order which allows recursive decomposition by means of isolated suborders with half the size of the starting order down to orders of a constant size (dismantable lattices as in~\cite{ReducibleLattices} are a class of orders with a similar property).
Then, the running time fulfills the recursion equation $T(|S|)\leq 3T(\frac{|S|}{2}+1)+p(|S|)$.
The critical exponent in this case is $\text{log}_2(3)<1.6$ which is less than the degree of $p$ according to the analysis of Algorithm~\ref{alg:compIsoBottleneck}.
Hence, under these (admittedly artificial and favorable) circumstances, Algorithm~\ref{alg:compIsoBottleneck} runs in $\bigo(p)$ time.

\section{Conclusion and Outlook}\label{sec:conclusion}

As main result, we introduced Algorithm~\ref{alg:counting} which can simplify counting of closures if the ordered set under consideration contains certain kinds of isolated suborders.
Naturally, the question arises whether more general or other structures than isolated suborders can be used with the same purpose.
Concepts which come to mind are bisimulations which are used also for simplification (in the sense of reducing the number of states) in model checking in~\cite{BaiKat} or model refinement in ~\cite{GlueMoeSintzAMAST2010}.
Also, it would be of interest to search for more special cases than the ones from Subsection~\ref{ssec:specialCases} wherefrom the Algorithm would clearly benefit.
Lastly, the algorithm is still awaiting being implemented.
Part of such an implementation would also be an optimal (i.e., linear time) algorithm for the computation of isolated suborders along the lines sketched at the end of Subsection~\ref{ssec:compISO}.
A combination of classical programming and specific systems like Mace4 (see~\cite{Mace4}) or RelView (see~\cite{RelView}) seems to be an interesting and promising approach.
Alas, for computing the number of closure systems on a powerset, the approach presented here will be of no help: it is easy to see that a powerset lattice contains no useful isolated suborders.\\[0.5cm]

\textbf{Acknowledgements: } The author is grateful to the anonymous reviewers for their critical but helpful comments which improved the readability of the paper and even lead to the discovery of a mathematical error in an earlier version.

\bibliographystyle{alphaurl}
\bibliography{HullsOrder}

\end{document}